\documentclass[12pt]{article}
\usepackage{amsfonts,amssymb,amsmath}
\usepackage[dvips]{epsfig}
\newtheorem{proposition}{Proposition}
\newenvironment{proof}{\noindent \textsc{Proof}\ }{\mbox{}\hfill $\Box$\\}

\newcounter{tempeq}
\begin{document}
\title{Spherical Harmonics $Y_{l}^{m}(\theta,\phi)$: Positive and Negative
Integer Representations of $su(1,1)$ for $l-m$ and $l+m$}
\author{H. Fakhri \thanks{Email: hfakhri@tabrizu.ac.ir} \\
\centerline{\small {\em Department of Theoretical Physics
and Astrophysics, Faculty of Physics, \,}}\\
{\small {\em University of Tabriz, P O Box 51666-16471, Tabriz, Iran
\,}}} \maketitle
\begin{abstract}
\noindent The azimuthal and magnetic quantum numbers of spherical
harmonics $Y_{l}^{m}(\theta,\phi)$ describe quantization
corresponding to the magnitude and $z$-component of angular momentum
operator in the framework of realization of $su(2)$ Lie algebra
symmetry. The azimuthal quantum number $l$ allocates to itself an
additional ladder symmetry by the operators which are written in
terms of $l$. Here, it is shown that simultaneous realization of the
both symmetries inherits the positive and negative $(l-m)$- and
$(l+m)$-integer discrete irreducible representations for $su(1,1)$
Lie algebra via the spherical harmonics on the sphere  as a compact
manifold. So, in addition to realizing the unitary irreducible
representation of $su(2)$ compact Lie algebra via the
$Y_{l}^{m}(\theta,\phi)$'s for a given $l$, we can also represent
$su(1,1)$ noncompact Lie algebra by spherical harmonics for given
values of $l-m$ and $l+m$.
\end{abstract}

\section{Introduction}
The set of principal, azimuthal, magnetic and spin quantum numbers
describe the unique quantum state of a single electron for any
system in which the potential depends only on the radial coordinate.
The labels $l$ and $m$ of the usual complex spherical harmonics
$Y_{l}^{m}(\theta,\phi)$ are the second and the third numbers of
this set. Spherical surface harmonics are an orthonormal set of
vibration solutions for eigenvalue equation of the Laplace-Beltrami
operator on the sphere $S^2$ as a compact Riemannian manifold. They
also form the wave functions which  represent the orbital angular
momentum operator ${\bf L}={\bf r}\times {\bf p}=-i{\bf r}\times
{\bf \nabla}$ ($\hbar=1$), and have a wide range of applications in
theoretical and applied physics
\cite{Macrobert,Infeld,Munakata,Varshalovich}. The three components
of the angular momentum operator, i.e. $L_x$, $L_y$ and $L_z$, are
Killing vector fields that generate the rotations about $x-$, $y-$
and $z-$axes, respectively. All spherical harmonics with the given
quantum number $l$ form a unitary irreducible representation of
$su(2)\cong so(3)$ Lie algebra. In fact, they can be seen as
representations of the $SO(3)$ symmetry group of rotations about a
point and its double-cover $SU(2)$. It can also be noted that the
spherical harmonics $Y_{l}^{m}(\theta,\phi)$ are just the
independent components of symmetric traceless tensors of rank $l$.
The properties of the spherical harmonics are well known, and may be
found in many texts and papers (for example see Refs.
\cite{Rose,Merzbacher,Beers,Dixon,Beig}).

In spite of the fact that the problem of quantization of particle
motion on a sphere is 80 years old, there still exist some open
questions concerning the symmetry properties of the bound states.
The aim of this work is to introduce new symmetries based on the
quantization of both azimuthal and magnetic numbers $l$ and $m$ of
the usual spherical harmonics $Y_{l}^{m}(\theta,\phi)$. In order to
provide the necessary background and also to attribute a
quantization relation for azimuthal quantum number $l$, here, we
present some basic facts about the spherical harmonics
\cite{Infeld,Munakata}. In section 2, with the application of
angular momentum operator, we review realization of the unitary
irreducible representations of $su(2)$ Lie algebra on the sphere in
terms of spherical harmonics by shifting $m$ only. In section 3, the
representations of the ladder symmetry with respect to azimuthal
quantum number $l$ are constructed in terms of a pair of ladder
operators, and its corresponding quantization relation is also
expressed as an operator identity originated from solubility in the
framework of supersymmetry and shape invariance theories. In section
4, these results are applied to show that $su(1,1)$ Lie algebra can
also be represented irreducibly by using spherical harmonics.
Finally, in section 5 we discuss the results and make some final
comments.

\section{The unitary irreducible representations of $su(2)$ Lie algebra via orbital
angular momentum operator} This section covers the standard and the
well-known formalism of $su(2)$ commutation relations in order to
encounter spherical harmonics. In what follows, we describe points
on $S^2$ using the parametrization $(x=r\sin\theta \cos\phi,
y=r\sin\theta \sin\phi, x=r\cos\theta)$ where $0\leq \theta <\pi$ is
the polar (or co-latitude) angle and $0\leq \phi<2\pi$ is the
azimuthal (or longitude) angle. For a given $l$ with the lower bound
$l\geq 0$, we define the $(2l+1)$-dimensional Hilbert space
\begin{eqnarray}
{\cal H}_l:=\mbox{span}\left\{Y_{l}^{m}(\theta, \phi)\right\}_{\hspace{1mm} -l \leq m\leq l}\,,
\end{eqnarray}
with the spherical harmonics as bases:
\begin{eqnarray}
&&\hspace{-15mm}Y_{l}^{m}(\theta, \phi)=\frac{(-1)^{m}}{2^{l} \,\Gamma(l+1)}\,
\sqrt{\frac{(2l+1)\Gamma(l+m+1)}{4\pi\Gamma(l-m+1) }}\left(\frac{e^{i\phi}}{\sin\theta}\right)^{m}
\left(\frac{1}{\sin\theta}\frac{d}{d\theta}\right)^{l-m}(\sin\theta)^{l}.
\end{eqnarray}
Also, the infinite dimensional Hilbert space ${\cal H}=L^2(S^2,d\Omega(\theta,\phi))$ is defined
as a direct sum of finite dimensional subspaces:
${\cal H}=\oplus_{l=0}^{+\infty}{\cal H}_l$. We must emphasize that
the bases of ${\cal H}$ are independent spherical harmonics with different
values for both indices $l$ and $m$. The
spherical harmonics as the bases of ${\cal H}$  constitute an
orthonormal set with respect to the following inner product over the sphere $S^2$:
\begin{eqnarray}
\int_{S^2} {Y_{l}^{m}}^*(\theta, \phi)
Y_{l^{\prime}}^{m^{\prime}}(\theta, \phi)d\Omega(\theta,\phi)=
\delta_{l\, l^{\prime}}\, \delta_{m\, m^{\prime}}.
\end{eqnarray}
Therefore, similar to the Fourier expansion, they can be used to expand
any arbitrary square integrable function of latitude and longitude angles.
$d\Omega(\theta,\phi)=d\cos\theta d\phi$ is the natural invariant measure (area) on the sphere $S^2$.
The following Proposition is an immediate consequence of
the raising and lowering relations of the index $m$ of the associated Legendre functions
\cite{Infeld,Varshalovich,Arfken}.
\begin{proposition}
Let us introduce three differential generators  $L_{+}$\,, $L_{-}$
and $L_{z}$ on the sphere $S^2$, corresponding to the orbital
angular momentum operator ${\bf L}$ as
\begin{eqnarray}
L_{\pm}=e^{\pm i\phi}\left(\pm \frac{\partial}{\partial \theta}
+i\cot\theta \frac{\partial}{\partial \phi} \right)\,,
\hspace{30mm} L_{z}=-i \frac{\partial}{\partial \phi}\,.
\end{eqnarray}
They satisfy the commutation relations of $su(2)$ Lie algebra as follows
\begin{eqnarray}
[L_{+}\, , L_{-} ]=2L_{z}\,, \hspace{30mm}[L_{z}\, , L_{\pm}]=\pm L_{\pm}\,.
\end{eqnarray}
$L_{z}$ is a self-adjoint operator, and two operators $L_{+}$ and
$L_{-}$ are Hermitian conjugate of each other with respect to the
inner product (3). Each of the Hilbert subspaces ${\cal H}_{l}$
realizes an $l-$integer unitary irreducible representation of
$su(2)$ Lie algebra as
\renewcommand\theequation{\arabic{equation}}
\setcounter{equation}{\value {tempeq}}
\setcounter{tempeq}{\value{equation}}
\renewcommand\theequation{\arabic{tempeq}\alph{equation}}
\setcounter{equation}{0} \addtocounter{tempeq}{6}
\begin{eqnarray}
&&L_{+}Y_{l}^{m-1}(\theta, \phi)=
\sqrt{(l-m+1)(l+m)}\,Y_{l}^{m}(\theta, \phi) \\
&&L_{-}Y_{l}^{m}(\theta, \phi)=
\sqrt{(l-m+1)(l+m)}\,Y_{l}^{m-1}(\theta, \phi) \\
&&L_{z}Y_{l}^{m}(\theta, \phi)=m Y_{l}^{m}(\theta, \phi)\,.
\end{eqnarray}
The Hilbert subspace ${\cal H}_{l}$ contains the lowest and highest
bases
\renewcommand\theequation{\arabic{equation}}
\setcounter{equation}{\value {tempeq}}
\begin{eqnarray}
Y_{l}^{\mp l}(\theta, \phi)=
\frac{\sqrt{\Gamma(2l+2)}}{\sqrt{\pi}\, 2^{l+1}\, \Gamma(l+1)}
\left(\sin\theta\right)^{l}\frac{e^{\mp i l \phi}}{(-1)^{\frac{l}{2}\mp \frac{l}{2}}}
\end{eqnarray}
with the lowest and highest weights $-l$ and $l$, respectively. They
are annihilated by the operators $L_{-}$ and $L_{+}$\,:
$L_{-}Y_{l}^{-l}(\theta, \phi)=0$ and $L_{+}Y_{l}^{l}(\theta,
\phi)=0$. Meanwhile, an arbitrary basis belonging to each of the
Hilbert subspaces ${\cal H}_{l}$ can be calculated by an algebraic
method as follows
\begin{eqnarray}
\hspace{-6mm}Y_{l}^{m}(\theta, \phi)=
\sqrt{\frac{\Gamma(l\mp m+1)}{\Gamma(2l+1)\Gamma(l\pm m+1)}}\left(L_{\pm}\right)^{l\pm m}
Y_{l}^{\mp l}(\theta, \phi)\,
\hspace{20mm}-l \leq m\leq l\,.
\end{eqnarray}
 Also, the Casimir operator corresponding to the generators (4),
i.e.
\begin{eqnarray}
{\bf L}^2_{\bf su(2)}
=L_{+} L_{-}+ L_{z}^2-L_{z}\,,
\end{eqnarray}
is a self-adjoint operator and has a $(2l+1)$-fold degeneracy on
${\cal H}_{l}$ as
\begin{eqnarray}
\hspace{10mm}{\bf L}^2_{\bf su(2)}Y_{l}^{m}(\theta, \phi)=
l(l+1)Y_{l}^{m}(\theta, \phi) \hspace{18mm}-l \leq m\leq l.
\end{eqnarray}
\end{proposition}
Obviously, the representation of the $su(2)$ Lie algebra in the
Hilbert space ${\cal H}$ via equations (6) is reducible.

A given unitary irreducible representation is characterized by the
index $l$. The spherical harmonics $Y_{l}^{m}(\theta, \phi)$, via
their $m$ index, describe quantization corresponding to commutation
relations of the three components of orbital angular momentum
operator. $L_{z}=-i \frac{\partial}{\partial \phi}$ is always a
Killing vector field which corresponds to an angular momentum about
the body-fixed $z$-axis. The Casimir operator ${\bf L}^2_{su(2)}$
along with the Cartan subalgebra generator $L_{z}$ describes the
Hamiltonian of a free particle on the sphere with dynamical symmetry
group $SU(2)$ and $(2l+1)$-fold degeneracy for the energy spectrum.
It must be emphasized that the spherical harmonics and their
mathematical structure, as given by Proposition 1, are playing a
more visible and important role in different branches of physics.
The Proposition 1 implies that the spherical harmonics are created
by orbital angular momentum operator. Schwinger has developed the
realization of this Proposition in the framework of creation and
annihilation operators of two-dimensional isotropic oscillator
\cite{Schwinger}.

\section{Ladder symmetry for the azimuthal quantum number $l$}
It is evident that simultaneous realization of laddering relations with respect to two different
parameters $l$ and $m$ of the associated Legendre functions gives
us the possibility to represent laddering relations with respect to the azimuthal quantum number $l$
of spherical harmonics. Representation of such ladder symmetry
by the spherical harmonics $Y_{l}^{m}(\theta, \phi)$ with the same $m$ but different $l$
induces a new splitting on the Hilbert space ${\cal H}$:
\begin{eqnarray}
{\cal H}=\oplus_{m=-\infty}^{+\infty}{\cal H}_m \hspace{20mm} \mbox{with} \hspace{5mm}
{\cal H}_m:=\mbox{span}\left\{Y_{l}^{m}(\theta, \phi)\right\}_{\hspace{1mm}  l\geq \left|m\right|}\,.
\end{eqnarray}
The following Proposition provides an alternative characterization
of the mathematical structure of spherical harmonics.
\begin{proposition}
Let us define two first-order differential operators on the sphere
$S^2$
\begin{eqnarray}
&&J_{\pm}(l)=\pm \sin\theta \frac{\partial}{\partial \theta}+l\cos\theta \,.
\end{eqnarray}
They satisfy the following operator identity in the framework of
shape invariance theory
\begin{eqnarray}
J_{-}(l+1)J_{+}(l+1)-J_{+}(l)J_{-}(l)=2l+1\,.
\end{eqnarray}
$J_{\pm}(l\pm 2)$ are the adjoint of the operators $J_{\mp}(l)$ with
respect to the inner product (3), i.e. we have
$J_{\mp}^{\dagger}(l)=J_{\pm}(l\pm 2)$. Each of the Hilbert
subspaces ${\cal H}_{m}$ realize the semi-infinite raising and
lowering relations with respect to $l$ as
\renewcommand\theequation{\arabic{equation}}
\setcounter{equation}{\value {tempeq}}
\setcounter{tempeq}{\value{equation}}
\renewcommand\theequation{\arabic{tempeq}\alph{equation}}
\setcounter{equation}{0} \addtocounter{tempeq}{8}
\begin{eqnarray}
&&\hspace{-18mm}J_{+}(l)Y_{l-1}^{m}(\theta, \phi)=
\sqrt{\frac{2l-1}{2l+1}(l-m)(l+m)}\,Y_{l}^{m}(\theta, \phi)
\hspace{17mm}l\geq \left|m\right|+1 \\
&&\hspace{-18mm}J_{-}(l)Y_{l}^{m}(\theta, \phi)=
\sqrt{\frac{2l+1}{2l-1}(l-m)(l+m)}\,Y_{l-1}^{m}(\theta, \phi)\hspace{17mm}l\geq \left|m\right| \,.
\end{eqnarray}
The lowest bases, i.e.
\renewcommand\theequation{\arabic{equation}}
\setcounter{equation}{\value {tempeq}}
\begin{eqnarray}
Y_{\pm m}^{m}(\theta, \phi)=\frac{(-1)^{-\frac{m}{2}\mp\frac{m}{2}}}{2^{\pm m}\Gamma(1\pm m)}
\sqrt{\frac{\Gamma(2\pm 2m)}{4\pi}}\, e^{im\phi}
\left(\sin\theta\right)^{\pm m}\,,
\end{eqnarray}
belonging to the Hilbert subspaces ${\cal H}_{m}$  with $m\geq 0$
and $m\leq 0$, are respectively annihilated by $J_-(m)$ and
$J_-(-m)$  as $J_-(m)Y_{m}^{m}(\theta, \phi)=0$ and
$J_-(-m)Y_{-m}^{m}(\theta, \phi)=0$. Meanwhile, an arbitrary basis
belonging to each of the Hilbert subspaces ${\cal H}_{m}$ with
$m\geq 0$ and $m\leq 0$, can be calculated by the algebraic method:
\begin{eqnarray}
&&\hspace{-17mm}Y_{l}^{m}(\theta, \phi)= \nonumber \\
&&\hspace{-17mm}\sqrt{\frac{(2l+1)\Gamma(1\pm 2m)}{(1\pm 2m)\Gamma(l-m+1)\Gamma(l+m+1)}}
\, J_{+}(l) J_{+}(l-1) \cdots J_{+}(1\pm m)Y_{\pm m}^{m}(\theta, \phi)\,.
\end{eqnarray}
\end{proposition}
\begin{proof}
The proof follows immediately from the raising and lowering
relations of the index $l$ of the associated Legendre functions
\cite{Infeld}.
\end{proof}
According to the $-l\leq m\leq +l$ limitation obtained from the
commutation relations of $su(2)$, $2l+1$ must be an odd and even
nonnegative integer for the orbital and spin angular momenta,
respectively. Although the relation (13) is identically satisfied
for any constant number $l$, however, it is represented only via the
nonnegative integers $l$ (odd positive integer values for $2l+1$) of
spherical harmonics $Y_{l}^{m}(\theta, \phi)$. This is an essential
difference with respect to the spin angular momentum. In fact, the
relation (13) distinguishes the orbital angular momentum from the
spin one. It also implies that the number of independent components
of spherical harmonics of a given irreducible representation $l$ of
$su(2)$ Lie algebra, i.e. $2l+1$, is derived by the shift operators
corresponding to the azimuthal quantum number $l$. If we take the
adjoint of the equation (13), we obtain
$J_{-}(l-1)J_{+}(l+3)-J_{+}(l+2)J_{-}(l-2)=2l+1$, which  is
identically satisfied. Thus, Proposition 2 presents a symmetry
structure, called ladder symmetry with respect to the azimuthal
quantum number $l$ of spherical harmonics. Note that, indeed, the
identical equality (13) has been originated from a brilliant theory
in connection with geometry and physics named supersymmetry. In
other words, although contrary to $L_+$ and $L_-$, the two operators
$J_{+}(l)$ and $J_{-}(l)$ do not contribute in a set of closed
commutation relations, however, the operator identity (13) for them
can be interpreted as a quantization relation in the framework of
shape invariance symmetry (for reviews about supersymmetric quantum
mechanics and shape invariance, see Refs.
\cite{Witten1,Witten2,Witten3,Alvarez,Turbiner,Khare}). Thus, the
operators $J_{+}(l)$ and $J_{-}(l)$ describe quantization of the
azimuthal quantum number $l$ which, in turn, lead to the
presentation of a different algebraic technique from (8), in order
to create the spherical harmonics $Y_{l}^{m}(\theta, \phi)$,
according to (16). Furthermore, spherical harmonics belonging to the
Hilbert subspaces ${\cal H}_{l}$ have parity $(-1)^{l}$, since ${\bf
L}$ commutes with the parity operator. Thus, the operators
$J_{+}(l)$ and $J_{-}(l)$ can be interpreted as the interchange
operators of parity: $J_{+}(l): {\cal H}_{l-1} \rightarrow {\cal
H}_{l}$ and $J_{-}(l): {\cal H}_{l} \rightarrow {\cal H}_{l-1}$.

\section{Positive and negative integer irreducible representations
of $u(1,1)$ for $l\mp m$} The laddering equations (6a) and (6b) as
well as (14a) and (14b), which describe shifting the indices $m$ and
$l$ separately, lead to the derivation of two new types of
simultaneous ladder symmetries with respect to the both azimuthal
and magnetic quantum numbers of spherical harmonics. Our proposed
ladder operators for simultaneous shift of $l$ and $m$ are
of first-order differential type, contrary to Ref. \cite{Infeld}.
They lead to a new perspective on the two quantum numbers $l$ and
$m$ in connection with realization of $u(1,1)$ (consequently, $su(1,1)$)
Lie algebra which in turn is accomplished by all spherical harmonics
$Y_{l}^{m}(\theta, \phi)$ with constant values for $l-m$ and $l+m$, separately.
First, it should be pointed out that the Hilbert space ${\cal H}$ can be
split into the infinite direct sums of infinite dimensional Hilbert
subspaces in two different ways as follows
\renewcommand\theequation{\arabic{equation}}
\setcounter{equation}{\value {tempeq}}
\setcounter{tempeq}{\value{equation}}
\renewcommand\theequation{\arabic{tempeq}\alph{equation}}
\setcounter{equation}{0} \addtocounter{tempeq}{3}
\begin{eqnarray}
&&\hspace{-10mm}{\cal H}=\left(\oplus_{j=0}^{\infty}{\cal H}_{d=2j+1}^{+}\right) {\oplus} \left(\oplus_{k=1}^{\infty} {\cal H}_{d=2k}^{+}\right) \nonumber \\
&&\hspace{50mm}\mbox{with}\hspace{3mm}\left\{\begin{array}{ll}
{\cal H}_{d=2j+1}^+=
\mbox{span}\left\{Y_{m+2j}^{m}(\theta, \phi)\right\}_{ m\geq -j }\\\\
{\cal H}_{d=2k}^+=
\mbox{span}\left\{Y_{m+2k-1}^{m}(\theta, \phi)\right\}_{ m\geq 1-k }
\end{array}\right. \\
&&\hspace{-10mm}{\cal H}=\left(\oplus_{j=0}^{\infty}{\cal H}_{s=2j+1}^-\right) {\oplus}
\left(\oplus_{k=1}^{\infty} {\cal H}_{s=2k}^-\right) \nonumber \\
&&\hspace{50mm}\mbox{with}\hspace{3mm}\left\{\begin{array}{ll}
{\cal H}_{s=2j+1}^-=
\mbox{span}\left\{Y_{-m+2j}^{m}(\theta, \phi)\right\}_{ m\leq j }\\\\
{\cal H}_{s=2k}^-=
\mbox{span}\left\{Y_{-m+2k-1}^{m}(\theta, \phi)\right\}_{ m\leq k-1 }\,.
\end{array}\right.
\end{eqnarray}
The constant values for the expressions $l-m$ and $l+m$ of spherical harmonics have been labeled by
$d-1$ and $s-1$, respectively.
\begin{proposition}
Let us define two new first-order differential operators on the
sphere $S^2$
\renewcommand\theequation{\arabic{equation}}
\setcounter{equation}{\value {tempeq}}
\begin{eqnarray}
K_{\pm}^{d}=e^{\pm i\phi}\left(\pm \cos\theta\frac{\partial}{\partial \theta}
+i\left(\frac{1}{\sin\theta}+\sin\theta\right)\frac{\partial}{\partial \phi}-\left(d-\frac{1}{2}\pm\frac{1}{2}\right)\sin\theta\right)\,.
\end{eqnarray}
They, together with the generators
$K_{z}=L_{z}=-i\frac{\partial}{\partial \phi}$ and $1$, satisfy the
commutation relations of $u(1,1)$ Lie algebra
\begin{eqnarray}
[K_{+}^{d}\, , K_{-}^{d}]=-8K_{z}-4d+2\,, \hspace{30mm}[K_{z}\, , K_{\pm}^{d}]=\pm K_{\pm}^{d}\,.
\end{eqnarray}
$K_{\pm}^{d\pm 2}$ are the adjoint of the operators $K_{\mp}^{d}$
with respect to the inner product (3), i.e. we have
${K_{\mp}^{d}}^{\dagger}=K_{\pm}^{d\pm 2}$. Each of the Hilbert
subspaces ${\cal H}_{d}^{+}$ realizes separately $(d-1)$-integer
irreducible positive representations of $u(1,1)$ Lie algebra
as\footnote{It must be pointed out that, by defining
$S_{z}^{d}:=L_{z}+\frac{d}{2}-\frac{1}{4}$ and
$S_{\pm}^{d}:=\frac{K_{\pm}^{d}}{2}$, the $u(1,1)$ Lie algebra (19)
can be considered as commutation relations corresponding to the
$su(1,1)$ Lie algebra: $[S_{+}^{d}\, , S_{-}^{d}]=-2S_{z}^{d}$ and
$[S_{z}^{d}\, , S_{\pm}^{d}]=\pm S_{\pm}^{d}$. This means that $1$
is a trivial center for the semisimple Lie algebra $u(1,1)$. In Ref.
\cite{Fakhri7}, a short review on the three different real forms
$h_4$, $u(2)$ and $u(1,1)$ of $gl(2,c)$ Lie algebra has been
presented. There, their differences in connection with the structure
constants and their representation spaces have also been pointed
out.}
\renewcommand\theequation{\arabic{equation}}
\setcounter{equation}{\value {tempeq}}
\setcounter{tempeq}{\value{equation}}
\renewcommand\theequation{\arabic{tempeq}\alph{equation}}
\setcounter{equation}{0} \addtocounter{tempeq}{3}
\begin{eqnarray}
&&\hspace{-23mm}K_{+}^{d}Y_{m+d-2}^{m-1}(\theta, \phi)=
\sqrt{\frac{2m+2d-3}{2m+2d-1}(2m+d-2)(2m+d-1)}\,Y_{m+d-1}^{m}(\theta, \phi) \\
&&\hspace{-23mm}K_{-}^{d}Y_{m+d-1}^{m}(\theta, \phi)=
\sqrt{\frac{2m+2d-1}{2m+2d-3}(2m+d-2)(2m+d-1)}\,Y_{m+d-2}^{m-1}(\theta, \phi)  \\
&&\hspace{-23mm}K_{z}Y_{m+d-1}^{m}(\theta, \phi)=m Y_{m+d-1}^{m}(\theta, \phi)\,.
\end{eqnarray}
Also, the Casimir operator corresponding to the generators
$K_{+}^{d}$, $K_{-}^{d}$ and $K_{z}$,
\renewcommand\theequation{\arabic{equation}}
\setcounter{equation}{\value {tempeq}}
\begin{eqnarray}
{\bf K^{d}}^2_{\bf u(1,1)}=K_{+}^{d}K_{-}^{d}-4K_{z}^{2}-2(2d-3)K_{z}\,,
\end{eqnarray}
has an infinite-fold degeneracy on the Hilbert subspace ${\cal
H}_{d}^{+}$ as
\begin{eqnarray}
{\bf K^{d}}^2_{\bf u(1,1)}Y_{m+d-1}^{m}(\theta, \phi)=(d-1)(d-2)Y_{m+d-1}^{m}(\theta, \phi)\,.
\end{eqnarray}
The Hilbert subspaces ${\cal H}_{d=2j+1}^+={\cal D}^{+}(-j)$ and
${\cal H}_{d=2k}^+={\cal D}^{+}(1-k)$ with $j$ and $k$ as
nonnegative and positive integers contain, respectively, the
following lowest bases
\renewcommand\theequation{\arabic{equation}}
\setcounter{equation}{\value {tempeq}}
\setcounter{tempeq}{\value{equation}}
\renewcommand\theequation{\arabic{tempeq}\alph{equation}}
\setcounter{equation}{0} \addtocounter{tempeq}{3}
\begin{eqnarray}
&&\hspace{-20mm}Y_{j}^{-j}(\theta, \phi)=\frac{1}{2^{j}\Gamma(j+1)}\sqrt{\frac{\Gamma(2j+2)}{4\pi}}\,e^{-ij\phi}
\left(\sin\theta\right)^j \\
&&\hspace{-20mm}Y_{k}^{1-k}(\theta, \phi)=\frac{1}{2^{k+\frac{1}{2}}\Gamma(k+1)}\sqrt{\frac{k\Gamma(2k+2)}{\pi}}\,e^{i(1-k)\phi}
\left(\sin\theta\right)^{k-1}\cos\theta\,.
\end{eqnarray}
They are annihilated as $K_{-}^{2j+1}Y_{j}^{-j}(\theta, \phi)=0$ and
$K_{-}^{2k}Y_{k}^{1-k}(\theta, \phi)=0$, and also have the lowest
weights $-j$ and $1-k$. Meanwhile, the arbitrary bases of the
Hilbert subspaces ${\cal H}_{d=2j+1}^+$ and ${\cal H}_{d=2k}^+$ can
be respectively calculated by the algebraic methods as
\renewcommand\theequation{\arabic{equation}}
\setcounter{equation}{\value {tempeq}}
\setcounter{tempeq}{\value{equation}}
\renewcommand\theequation{\arabic{tempeq}\alph{equation}}
\setcounter{equation}{0} \addtocounter{tempeq}{1}
\begin{eqnarray}
&&Y_{m+2j}^{m}(\theta, \phi)=
\frac{\left(K_{+}^{2j+1}\right)^{m+j}Y_{j}^{-j}(\theta, \phi)}
{\sqrt{\frac{(2j+1)\Gamma(2m+2j+1)}{2m+4j+1}}} \hspace{30mm}m\geq -j  \\
&&Y_{m+2k-1}^{m}(\theta, \phi)=
\frac{\left(K_{+}^{2k}\right)^{m+k-1}Y_{k}^{1-k}(\theta, \phi)}
{\sqrt{\frac{(2k+1)\Gamma(2m+2k)}{2m+4k-1}}} \hspace{20mm}m\geq 1-k\,.
\end{eqnarray}
\end{proposition}
\begin{proof} The relations (18) and (20a,b) can be followed from the realization of
laddering relations with respect to both azimuthal and magnetic
quantum numbers $l$ and $m$, simultaneously and agreeably.
It is sufficient to consider that two new differential operators
\renewcommand\theequation{\arabic{equation}}
\setcounter{equation}{\value {tempeq}}
\begin{eqnarray}
A_{\pm,\pm}(l):=\pm[L_{\pm}, J_{\pm}(l)]=e^{\pm i\phi}\left(\pm \cos\theta\frac{\partial}{\partial \theta}
+\frac{i}{\sin\theta}\frac{\partial}{\partial \phi}-l\sin\theta\right)\,,
\end{eqnarray}
satisfy the simultaneous laddering relations with respect to $l$ and $m$ as
\renewcommand\theequation{\arabic{equation}}
\setcounter{equation}{\value {tempeq}}
\setcounter{tempeq}{\value{equation}}
\renewcommand\theequation{\arabic{tempeq}\alph{equation}}
\setcounter{equation}{0} \addtocounter{tempeq}{2}
\begin{eqnarray}
&&\hspace{-7mm}A_{+,+}(l)Y_{l-1}^{m-1}(\theta, \phi)=
\sqrt{\frac{2l-1}{2l+1}(l+m-1)(l+m)}\,Y_{l}^{m}(\theta, \phi) \\
&&\hspace{-7mm}A_{-,-}(l)Y_{l}^{m}(\theta, \phi)=
\sqrt{\frac{2l+1}{2l-1}(l+m-1)(l+m)}\,Y_{l-1}^{m-1}(\theta, \phi)\,.
\end{eqnarray}
The relations (26a) and (26b) are obtained from (6a), (6b), (14a)
and (14b). The relation (19) and (20c) are directly followed. The
adjoint relation between the operators can be easily checked by
means of the inner product (3). The commutativity of operators
$K_{+}^{d}$, $K_{-}^{d}$ and $K_{z}$ with ${\bf K^{d}}^2_{\bf
u(1,1)}$ is resulted from (19). The eigenequation (22) follows
immediately from the representation relations (20). The relation
(20b) implies that $Y_{j}^{-j}(\theta, \phi)$ and
$Y_{k}^{1-k}(\theta, \phi)$ are the lowest bases for the Hilbert
subspaces ${\cal H}_{2j+1}^{+}$ and ${\cal H}_{2k}^{+}$,
respectively. Then, with repeated application of the raising
relation (20a), one may obtain the arbitrary representation bases of
$u(1,1)$ Lie algebra as (24a) and (24b).
\end{proof}
Although the commutation relations (19) are not closed with respect
to taking the adjoint, however, their adjoint relations
$[K_{+}^{d+2}\, , K_{-}^{d-2}]=-8K_{z}-4d+2$ and $[K_{z}\, ,
K_{\mp}^{d\mp 2}]=\mp K_{\mp}^{d\mp 2}$ are identically satisfied.
\begin{proposition} Let us define two new first-order differential
operators on the sphere $S^2$ as
\renewcommand\theequation{\arabic{equation}}
\setcounter{equation}{\value {tempeq}}
\begin{eqnarray}
I_{\pm}^{s}=e^{\pm i\phi}\left(\pm \cos\theta\frac{\partial}{\partial \theta}
+i\left(\frac{1}{\sin\theta}+\sin\theta\right)\frac{\partial}{\partial \phi}+\left(s-\frac{1}{2}\mp\frac{1}{2}\right)\sin\theta\right)\,.
\end{eqnarray}
They, together with the generators
$I_{z}=L_{z}=-i\frac{\partial}{\partial \phi}$ and $1$ satisfy the
commutation relations of $u(1,1)$ Lie algebra as
\begin{eqnarray}
[I_{+}^{s}\, , I_{-}^{s}]=-8I_{z}+4s-2\,, \hspace{30mm}[I_{z}\, , I_{\pm}^{s}]=\pm I_{\pm}^{s}\,.
\end{eqnarray}
$I_{\pm}^{s \mp 2}$ are the adjoint of the operators $I_{\mp}^{s}$
with respect to the inner product (3), i.e. we have
${I_{\mp}^{s}}^{\dagger}=I_{\pm}^{s \mp 2}$. Each of the Hilbert
subspaces ${\cal H}_{s}^{-}$ realize separately $(s-1)$-integer
irreducible positive representations of $u(1,1)$ Lie algebra as
\renewcommand\theequation{\arabic{equation}}
\setcounter{equation}{\value {tempeq}}
\setcounter{tempeq}{\value{equation}}
\renewcommand\theequation{\arabic{tempeq}\alph{equation}}
\setcounter{equation}{0} \addtocounter{tempeq}{3}
\begin{eqnarray}
&&\hspace{-23mm}I_{+}^{s}Y_{-m+s}^{m-1}(\theta, \phi)=
\sqrt{\frac{-2m+2s+1}{-2m+2s-1}(-2m+s)(-2m+s+1)}\,Y_{-m+s-1}^{m}(\theta, \phi)  \\
&&\hspace{-23mm}I_{-}^{s}Y_{-m+s-1}^{m}(\theta, \phi)=
\sqrt{\frac{-2m+2s-1}{-2m+2s+1}(-2m+s)(-2m+s+1)}\,Y_{-m+s}^{m-1}(\theta, \phi)  \\
&&\hspace{-23mm}I_{z}Y_{-m+s-1}^{m}(\theta, \phi)=m Y_{-m+s-1}^{m}(\theta, \phi)\,.
\end{eqnarray}
Also, the Casimir operator corresponding to the generators
$I_{+}^{s}$, $I_{-}^{s}$ and $I_{z}$,
\renewcommand\theequation{\arabic{equation}}
\setcounter{equation}{\value {tempeq}}
\begin{eqnarray}
{\bf I^{s}}^2_{\bf u(1,1)}=I_{+}^{s}I_{-}^{s}-4I_{z}^{2}+2(2s+1)I_{z}\,,
\end{eqnarray}
has an infinite-fold degeneracy on the Hilbert subspace ${\cal
H}_{s}^{-}$ as
\begin{eqnarray}
{\bf I^{s}}^2_{\bf u(1,1)}Y_{-m+s-1}^{m}(\theta, \phi)=s(s+1)Y_{-m+s-1}^{m}(\theta, \phi)\,.
\end{eqnarray}
The Hilbert subspaces ${\cal H}_{s=2j+1}^-={\cal D}^{-}(j)$ and
${\cal H}_{s=2k}^-={\cal D}^{-}(k-1)$ with $j$ and $k$ as
nonnegative and positive integers contain, respectively, the
following highest bases
\renewcommand\theequation{\arabic{equation}}
\setcounter{equation}{\value {tempeq}}
\setcounter{tempeq}{\value{equation}}
\renewcommand\theequation{\arabic{tempeq}\alph{equation}}
\setcounter{equation}{0} \addtocounter{tempeq}{3}
\begin{eqnarray}
&&\hspace{-23mm}Y_{j}^{j}(\theta, \phi)=
\frac{(-1)^{j}}{2^{j}\Gamma(j+1)}\sqrt{\frac{\Gamma(2j+2)}{4\pi}}\,e^{ij\phi}
\left(\sin\theta\right)^j \\
&&\hspace{-23mm}Y_{k}^{k-1}(\theta, \phi)=\frac{(-1)^{k-1}}{2^{k-\frac{1}{2}}\Gamma(k+1)}\sqrt{\frac{(2k+1)\Gamma(2k)}{2\pi}}\,e^{i(k-1)\phi}
\left(\sin\theta\right)^{k-1}\cos\theta\,.
\end{eqnarray}
They are annihilated as $I_{+}^{2j+1}Y_{j}^{j}(\theta, \phi)=0$ and
$I_{+}^{2k}Y_{k}^{k-1}(\theta, \phi)=0$, and also have the highest
weights $j$ and $k-1$. Meanwhile, the arbitrary bases of the Hilbert
subspaces ${\cal H}_{s=2j+1}^-$ and ${\cal H}_{s=2k}^-$ can be
respectively calculated by the algebraic methods as
\renewcommand\theequation{\arabic{equation}}
\setcounter{equation}{\value {tempeq}}
\setcounter{tempeq}{\value{equation}}
\renewcommand\theequation{\arabic{tempeq}\alph{equation}}
\setcounter{equation}{0} \addtocounter{tempeq}{1}
\begin{eqnarray}
&&\hspace{-23mm}Y_{2j-m}^{m}(\theta, \phi)=
\frac{\left(I_{-}^{2j+1}\right)^{j-m}Y_{j}^{j}(\theta, \phi)}
{\sqrt{\frac{(2j+1)\Gamma(2j-2m+1)}{4j-2m+1}}} \hspace{30mm}m\leq j  \\
&&\hspace{-23mm}Y_{2k-m-1}^{m}(\theta, \phi)=
\frac{\left(I_{-}^{2k}\right)^{k-m-1}Y_{k}^{k-1}(\theta, \phi)}
{\sqrt{\frac{(2k+1)\Gamma(2k-2m)}{4k-2m-1}}} \hspace{20mm}m\leq k-1\,.
\end{eqnarray}
\end{proposition}
\begin{proof} The proof is quite similar to the proof of the Proposition 3. So, we
have to take into account that the two new differential operators
\renewcommand\theequation{\arabic{equation}}
\setcounter{equation}{\value {tempeq}}
\begin{eqnarray}
A_{\mp,\pm}(l):=\mp[L_{\pm}, J_{\mp}(l)]=e^{\pm i\phi}\left(\pm \cos\theta\frac{\partial}{\partial \theta}
+\frac{i}{\sin\theta}\frac{\partial}{\partial \phi}+l\sin\theta\right)\,,
\end{eqnarray}
are represented by spherical harmonics whose corresponding laddering equations shift both the
azimuthal and magnetic quantum numbers $l$ and $m$ simultaneously and inversely:
\renewcommand\theequation{\arabic{equation}}
\setcounter{equation}{\value {tempeq}}
\setcounter{tempeq}{\value{equation}}
\renewcommand\theequation{\arabic{tempeq}\alph{equation}}
\setcounter{equation}{0} \addtocounter{tempeq}{2}
\begin{eqnarray}
&&\hspace{-20mm}A_{-,+}(l+1)Y_{l+1}^{m-1}(\theta, \phi)=
\sqrt{\frac{2l+3}{2l+1}(l-m+1)(l-m+2)}\,Y_{l}^{m}(\theta, \phi) \\
&&\hspace{-20mm}A_{+,-}(l+1)Y_{l}^{m}(\theta, \phi)=
\sqrt{\frac{2l+1}{2l+3}(l-m+1)(l-m+2)}\,Y_{l+1}^{m-1}(\theta,
\phi)\,.
\end{eqnarray}
\end{proof}
Here, again the adjoint of commutation relations (28) becomes
$[I_{+}^{s-2}\, , I_{-}^{s+2}]=-8I_{z}+4s-2$ and $[I_{z}\, , I_{\mp}^{s\pm 2}]=\mp I_{\mp}^{s\pm 2}$,
which are identically satisfied.

Thus, all unitary and irreducible representations of $su(2)$ of
dimensions $2l+1$ with the nonnegative integers $l$ can carry the new kind of
irreducible representations for $u(1,1)$.
The new symmetry structures presented in the two recent Propositions, the
so-called positive and negative discrete representations of
$u(1,1)$, in turn, describe the simultaneous quantization of the
azimuthal and magnetic quantum numbers.
Therefore, the Hilbert space of all spherical harmonics
not only represent compact Lie algebra $su(2)$ by ladder operators
shifting $m$ for a given $l$, but also represent the noncompact Lie
algebra $u(1,1)$ by simultaneous shift operators of both quantum
labels $l$ and $m$ for given values $l-m$ and $l+m$.

\section{Concluding remarks}
For a given azimuthal quantum number $l$, quantization of the
magnetic number $m$ is customarily accomplished by representing the
operators $L_{+}$, $L_{-}$ and $L_{3}$ on the sphere with the
commutation relations $su(2)$ compact Lie algebra, in a
$(2l+1)$-dimensional Hilbert subspace ${\cal H}_{l}$. Furthermore,
for a given magnetic quantum number $m$, quantization of the
azimuthal number $l$ is accomplished by representing the operators
$J_{+}(l)$ and $J_{-}(l)$ on the sphere $S^2$ with the identity
relation (13), in an infinite-dimensional Hilbert subspace ${\cal
H}_{m}$.

Dealing with these issues together,
simultaneous quantization of both azimuthal and magnetic numbers $l$ and $m$
is accomplished by representing two bunches of operators $\{K_{+}^{d}, K_{-}^{d}, K_{3}, 1\}$
and $\{I_{+}^{s}, I_{-}^{s}, I_{3}, 1\}$ on the sphere with their corresponding commutation
relations of $u(1,1)$ noncompact Lie algebra, in the infinite-dimensional Hilbert subspaces
${\cal H}_{d}^{+}$ and ${\cal H}_{s}^{-}$, respectively.
For given values $d=l-m+1$ and $s=l+m+1$, they are independent of each other, the so-called
positive and negative $(l-m)$- and $(l+m)$-integer irreducible representations, respectively.
As the spherical harmonics are generated
from $Y_{l}^{\mp l}(\theta, \phi)$ by the operators $L_{\pm}$, they
are also generated from $Y_{j}^{-j}(\theta, \phi)$ and $Y_{k}^{1-k}(\theta, \phi)$ by
$K_{+}^{2j+1}$ and $K_{+}^{2k}$, as well as from $Y_{j}^{j}(\theta, \phi)$ and
$Y_{k}^{k-1}(\theta, \phi)$ by $I_{-}^{2j+1}$ and $I_{-}^{2k}$, respectively.
Therefore, not only  $Y_{l}^{m}(\theta, \phi)$'s with the given value for $l$ represent $su(2)$
Lie algebra, but also $Y_{l}^{m}(\theta, \phi)$'s with the given values for subtraction and summation
of the both quantum numbers $l$ and $m$ represent separately $u(1,1)$ (hence, $su(1,1)$) Lie algebra as well.
In other words, two different real forms of $sl(2,c)$ Lie algebra, i.e. $su(2)$ and $su(1,1)$,
are represented by the space of all spherical harmonics $Y_{l}^{m}(\theta, \phi)$.
This happens because the quantization of both quantum numbers $l$ and $m$ are considered jointly.
Indeed, we have
$$ \mbox{Propositions 1 and 2} \Longleftrightarrow \mbox{Propositions 3 and 4.}$$

We point out that the idea of this manuscript may find interesting
applications in quantum devices. For instance, coherent states of
the $SU(1,1)$ noncompact Lie group have been defined by Barut and
Girardello as eigenstates of the ladder operators \cite{Barut}, and
by Perelomov as the action of the displacement operator on the
lowest and highest bases \cite{Perelomov1,Perelomov2}. So, our
approach to the representation of $su(1,1)$ noncompact Lie algebra
provides the possibility of constructing two different types of
coherent states of $su(1,1)$ on compact manifold $S^2$
\cite{Fakhri5}. Also, realization of the additional symmetry named
$su(2)$ Lie algebra for Landau levels and bound states of a free
particle on noncompact manifold $AdS_2$ can be found based on the
above considerations.


\begin{thebibliography}{99}

\bibitem{Macrobert} T.M. Macrobert, {\em Spherical Harmonics: An Elementary Treatise on Harmonic
Functions with Applications}, Methuen \& Co. Ltd, London (1947).

\bibitem{Infeld} L. Infeld and T.E. Hull,
Rev. Mod. Phys. {\bf 23} (1), 21-68 (1951).

\bibitem{Munakata} Y. Munakata,
Commun. Math. Phys. {\bf 9} (1), 18-37 (1968).

\bibitem{Varshalovich} D.A. Varshalovich,  A.N. Moskalev  and V.K. Khersonsky,  {\em Quantum Theory of
Angular Momentum: Irreducible Tensors, Spherical Harmonics, Vector Coupling Coefficients, 3nj
Symbols}, World Scientiffc, Singapore (1989).

\bibitem{Rose} M.E. Rose,  {\em Elementary Theory of Angular Momentum},  Wiley, New York (1957).

\bibitem{Merzbacher} E. Merzbacher, {\em Quantum Mechanics},  Wiley, New York (1970).

\bibitem{Beers} B.L. Beers  and A.J. Dragt,
J. Math. Phys. {\bf 11} (8), 2313-2328 (1970).

\bibitem{Dixon} J.M. Dixon and  R. Lacroix, J. Phys. A: Math., Nucl. Gen. {\bf 6} (8), 1119-1128 (1973).

\bibitem{Beig} R. Beig,
J. Math. Phys. {\bf 26} (4), 769-770 (1985).

\bibitem{Arfken} G.B. Arfken, {\em Mathematical Methods for Physicists},
3rd ed., Academic Press, New York (1985).

\bibitem{Schwinger} J. Schwinger, {\em Quantum Theory of Angular Momentum},
Academic Press, New York (1952).

\bibitem{Witten1} E. Witten, Nucl. Phys. {\bf B188} (3-5), 513-554 (1981).

\bibitem{Witten2} E. Witten, Nucl. Phys. {\bf B202} (2), 253-316 (1982).

\bibitem{Witten3} E. Witten, J. Diff. Geo. {\bf 17} (4), 661-692 (1982).

\bibitem{Alvarez} L. Alvarez-Gaum\'e,
Commun. Math. Phys. {\bf 90} (2), 161-173 (1983).

\bibitem{Turbiner} A.V. Turbiner, Commun. Math. Phys. {\bf 118} (3), 467-474 (1988).

\bibitem{Khare} F. Cooper, A. Khare and U. Sukhatme, Phys. Rep. {\bf 251} (5-6),  267-385 (1995).

\bibitem{Fakhri7} H. Fakhri and A. Chenaghlou, J. Phys. A: Math. Theor. {\bf 40} (21), 5511-5523 (2007).

\bibitem{Barut} A.O. Barut and L. Girardello, Commun. Math. Phys. {\bf 21} (1), 41-55 (1971).

\bibitem{Perelomov1} A.M. Perelomov,
Commun. Math. Phys. {\bf 26} (3), 222-236 (1972).

\bibitem{Perelomov2} A.M. Perelomov, {\em Generalized Coherent States and Their Applications},
Springer-Verlag, Berlin (1986).

\bibitem{Fakhri5} H. Fakhri and A. Dehghani, J. Math. Phys. {\bf 50} (5), 052104 (2009).




\end{thebibliography}
\end{document}